\documentclass{llncs}
\usepackage{amsfonts,amssymb,amsmath,latexsym,ae,aecompl}
\usepackage{epsfig}
\usepackage{arcs}
\usepackage{enumerate}
\usepackage{graphicx}
\usepackage{algorithmic}
\usepackage{float,epsfig, floatflt,here}
\usepackage{psfrag}
\usepackage{pstricks, pst-all, pstricks-add}
\usepackage{graphicx}
\usepackage{caption}
\usepackage{epsfig}
\usepackage{subfig}
\usepackage{a4wide}
\usepackage{gensymb}
\usepackage{epstopdf}

\usepackage{enumerate}
\usepackage{url}
\usepackage{psfrag}
\usepackage{caption}
\usepackage{float}
\usepackage{algorithmic}
\usepackage{algorithm}
\usepackage{graphics}
\usepackage{graphicx}

\newcommand{\cE}{{\mathcal E}}

\newcommand{\cI}{{\mathcal I}}

\newcommand{\cS}{{\mathcal S}}

\newcommand{\cC}{{\mathcal C}}

\newcommand{\cL}{{\mathcal L}}
\newcommand{\cX}{{\mathcal X}}

\begin{document}
\title{Approximation Algorithms for Barrier Sweep Coverage
\thanks{A preliminary version of this paper appears in Proc. 6th International Conference on Communication System \& Networks (COMSNET'14)}
}

\author{Barun Gorain \and Partha Sarathi Mandal}
\institute{\small   Department of Mathematics\\Indian Institute of Technology Guwahati, India\\
\email{\{b.gorain,psm\}@iitg.ernet.in}}

\maketitle
\begin{abstract}
Time-varying coverage, namely sweep coverage is a recent development in the area of wireless sensor networks,
where a small number of mobile sensors sweep or monitor comparatively large number of locations periodically.
In this article we study barrier sweep coverage with mobile sensors where the barrier is considered as
a finite length continuous curve on a plane. The coverage at every point on the curve is time-variant.
We propose an optimal solution for sweep coverage of a finite length continuous curve.
Usually energy source of a mobile sensor is battery with limited power, so energy restricted sweep
coverage is a challenging problem for long running applications.
We propose an energy restricted sweep coverage problem where every mobile sensors must visit an energy source frequently to recharge or replace its battery. We propose a $\frac{13}{3}$-approximation algorithm for this problem.
The proposed algorithm for multiple curves achieves the best possible approximation factor 2 for
a special case. We propose a 5-approximation algorithm for the general problem.
As an application of the barrier sweep coverage problem for a set of line segments, we formulate
a data gathering problem. In this problem a set of mobile sensors is arbitrarily monitoring the line
segments one for each. A set of data mules periodically collects the monitoring data from the set of
mobile sensors. We prove that finding the minimum number of data mules to collect data periodically from
every mobile sensor is NP-hard and propose a 3-approximation algorithm to solve it.
\end{abstract}


\section{Introduction}

Coverage is a widely studied research area and one of the most important problem
in wireless sensor networks (WSNs) for monitoring, surveillance, etc.
Based on the subject to be covered by a set of sensors,
it is classified into three categories, such as point coverage, area coverage and
barrier coverage. In point coverage \cite{Gu09,LU05}, a set of points are covered,
whereas in area coverage \cite{CardeiM05,Saravi09,WangCP03}, all points inside a bounded
area are covered. But in barrier coverage, barriers of sensors are required for an
appropriate model of coverage for applications like detecting intruders
when they cross borders or detecting spread of pollutants, chemicals when sensors are
deployed around critical regions.
In most of barrier coverage literatures \cite{Chen07,Kumar05,Liu08,Saipulla09,YangQ09}
static sensors are usually used for continuous monitoring the borders or boundaries.
%Coverage with static sensors, referred to as {\it static coverage} in this paper.
But there are applications \cite{Du2010}, where time-variant coverage at every point on
a boundary is sufficient instead of monitoring all along.
For these kind of applications, deployment of static sensors at the boundary is not
cost effective in terms of resource utilization for periodic monitoring every point
on the boundary. The time-variant coverage can be solved efficiently by utilizing less number of resources
i.e., mobile sensors with appropriate movement strategy.
The cost involvement aspect of this solution is mobility and storage capacity of the
mobile sensors.
This type of coverage problems, where time-variant coverage is sufficient for
periodic petrol inspections are termed as {\it sweep coverage}.
In point sweep coverage problem \cite{Hung10,Du2010,Gorain2014,Li11,Xi09}, a given set of discrete points are monitored by a set of mobile sensors at least once within a given period of time. This time period is termed as {\it sweep period}. The primary motivation of these sweep coverage problems is to utilize minimum number of mobile sensors for achieving the goal. But finding minimum number of mobile sensors for the sweep coverage of a given
set of discrete points on a plane is NP-hard \cite{Li11}. The area sweep coverage problem is
introduced in the article \cite{Gorain2014}, where it is proved that the problem is NP-complete.
In this article, we formulate different variation of barrier sweep coverage problems for covering finite length curves on a plane.

\subsection{Contribution}
In this paper, our contributions on sweep coverage problems are as follows:
\begin{itemize}
\item We introduce barrier sweep coverage problems for covering finite length curves on a plane.
We solve the problem optimally with minimum number of mobile sensors for a finite length curve.
\item We define an energy restricted barrier sweep coverage problem and propose a $\frac{13}{3}$-approximation algorithm for a finite length curve.
 \item  A $2$-approximation algorithm is proposed for solving the sweep coverage problem for multiple curves for a special case where each mobile sensor visits all points of each curve. For the general version of the problem, a 5-approximation algorithm is proposed.
 \item We formulate a data gathering problem by a set of data mules and a $3$-approximation algorithm is proposed to solve the NP-hard problem.
 \item Performance of the proposed algorithms are investigated through simulation for multiple finite length curves.
\end{itemize}

\subsection{Related Work}
The concept of sweep coverage initially came from the context of robotics \cite{Batalin02}.
In \cite{Batalin02}, the authors considered a dynamic sensor coverage problem using mobile robots in absence of global localization information.
The sensors are mobilized by mounting them on the mobile robots. The robots explore an unexplored area by deploying small communication beacons.
The robots decide direction of movements during the exploration using local markers with the beacons.
Recently, several works
\cite{Hung10,Du2010,Li11,ShuCZZ14,Chao11,Xi09,Yang13}  on sweep coverage are proposed in WSNs.
Most of the works focused on designing suitable heuristics.
In \cite{Xi09}, the authors considered a network consisting static and mobile sensors.
Two different problems are considered in that paper. In the first problem,
objective is to minimize number of mobile sensors that can guarantee data download from every static sensor in a given time period with high probability. In the second problem,
objective is to guide mobile sensors for moving towards static sensors without any centralized control.
In the first heuristic ({\it MinExpand}), mobile sensors move in the same path in every time period.  In the second heuristic ({\it OSweep}), the mobile sensors move in different
paths in different time periods. Hung et al. \cite{Hung10} considered a sweep coverage problem where
nonuniform deployment of a set of PoIs is made over an area of interest. The area is divided into
smaller sub-areas. Then mobile sensors are deployed over the sub-areas, one for each, to
guarantee sweep coverage of all the PoIs in the respective sub-areas. Due to unequal number of
PoIs in different sub-areas, sweep periods (patrolling times) of the mobile sensors may not be same.
Objective of the proposed heuristic is to make the patrolling time approximately same for all mobile sensors.
In \cite{ShuCZZ14}, the authors considered a problem where sweep periods of the PoIs are different. A scheme is proposed based on periodic vehicle routing problem to minimize number of unnecessary visits of a PoI by a mobile sensor. To extend lifetime of sweep coverage, Yang et al. \cite{Yang13}
utilized base station as a power source for periodical refueling or replacing battery of the mobile sensors.
The authors proposed two heuristics with one base station and multiple base stations,
respectively.
Hardnass of the sweep coverage problem is studied by Li et al. \cite{Li11} theoretically. The authors
proved that finding minimum number of mobile sensors to sweep cover a set of PoIs is NP-hard. It is proved that the problem cannot be approximated less than a factor of 2, unless P=NP.
A $(2+\epsilon)$-approximation and a 3-approximation algorithms are proposed to solve the problem.
The authors remarked on impossibility of design distributed local algorithm to guarantee sweep
coverage of all PoIs, i.e., a mobile sensor cannot locally determine whether all PoIs are sweep
covered without global information. But there is a flaw in the approximation algorithms.
We, in \cite{Gorainsss14}, remarked on the flaw of the 3-approximation algorithm \cite{Li11} and proposed corrected algorithm keeping same approximation factor to guarantee sweep coverage for the given set of PoIs. If the sweep periods of the PoIs are different, a $O(\log \rho)$-approximation algorithm we proposed where $\rho$ is the ratio of the maximum and minimum sweep periods. An inapproximability result is also established when the speed of the mobile sensors are not necessarily same.
In \cite{Gorain2014}, a
2-approximation algorithm for a special case of point sweep coverage problem we proposed, which is the best possible approximation factor according to the result in \cite{Li11}. A distributed 2-approximation algorithm we proposed for the point sweep coverage problem, where static sensors are deployed at the PoIs and those are considered as PoIs. The static sensors communicate among them self
through message to find initial deployment locations of the mobile sensors.
The area sweep coverage problem is formulated and proved that the problem is NP-complete.
A $2\sqrt{2}$ approximation algorithm is proposed to solve the problem for a square region. The area sweep coverage
problem for arbitrary bounded region is also investigated in that paper.
An energy efficient sweep coverage problem is proposed in \cite{Gorain015},
where a set of static and mobile sensors are used for sweep coverage of a set of discrete points. The objective is to minimize
the total energy consumption per unit time by the set of sensors guaranteeing the required sweep coverage. An 8-approximation
algorithm is proposed to solve the problem.

There are similar patrolling problems \cite{Czyzowicz2011,Dumitrescu2014,Kawamura2014} like sweep coverage with different objectives. Objective of the problems is to minimize time between two consecutive visits
of any point while monitoring a given road network or boundary of a region by a set of mobile agents
having different speeds. In \cite{Czyzowicz2011}, the authors proposed two strategies;  partition based strategy
and cycle based strategy to obtain movement schedules of the mobile agents. The authors proved that the strategy
obtains optimal solution when number of agents is less than or equal to 2 for partition strategy and
number of agents is less than or equal to 4 for cycle strategy. Kawamura et al. \cite{Kawamura2014}
proved that the partition strategy proposed in the paper \cite{Czyzowicz2011} achieves optimal solution
for number of agents less than or equal to 3.

The rest of the paper is organized as follows. The problem definitions are given in section \ref{sec:probDefination}.
Barrier sweep coverage for single curve and multiple curves are discussed in section \ref{sec:single} and section
\ref{sec:multiple} respectively. A data gathering problem by data mules is formulated and discussed in section \ref{sec:MDMDG}.
Finally, conclude and future works are given in section \ref{sec:concl}.

\section{Problem Definitions}\label{sec:probDefination}

Let $\cC$ be a finite length curve on a 2D plane. $\cC$ is said to be {\it covered} by a set of sensors
if and only if each point on $\cC$ is covered by at least one sensor. Based on the above coverage metric, the
definitions of barrier sweep coverage problems are given below.
\begin{definition} \rm({\bf $t-$barrier sweep coverage}\rm)
Let $\cC$ be any finite length curve on a 2D plane and $\cS=\{s_1,s_2,\cdots,s_m\}$ be a set of mobile sensors.
$\cC$ is said to be $t$-barrier sweep covered if and only if each point of $\cC$ is visited by at least one mobile sensor in
every time period $t$.
\end{definition}

\begin{definition}
\rm({\bf Barrier sweep coverage problem for single finite length curve}\rm) Let $\cC$ be a finite length curve
and $\cS=\{s_1,s_2,\cdots,s_m\}$ be a set of mobile sensors. For given $t>0$ and $v>0$, find the minimum number of
mobile sensors with uniform speed $v$ such that all points of $\cC$ are visited by all mobile sensors and $\cC$ is
$t$-barrier sweep covered.
\end{definition}

\begin{definition}
\rm({\bf Barrier sweep coverage problem for multiple finite length curves \rm(BSCMC\rm)}\rm) Let $\cX=\{\cC_1,\cC_2,\cdots,\cC_n\}$
be a set of $n$ finite length curves and $\cS=\{s_1,s_2,\cdots,s_m\}$ be a set of mobile sensors. For given  $t>0$ and $v>0$,
find the minimum number of mobile sensors with uniform speed $v$ such that each $\cC_i$ for $i=1,2,\cdots, n$ is $t$-barrier sweep covered.
\end{definition}

In general sensors are equipped with limited battery power. In order to continue sweep coverage for long time,
each mobile sensor must visits an energy source to recharge or replace its battery. Let $T$ be the maximum travel
time for a mobile sensor starting with full powered battery maintaining uniform speed $v$ till battery power goes off.
Let $e$ be an energy source on the plane and every mobile sensor can recharge or replace its battery by visiting $e$.
We define energy restricted barrier sweep coverage problem for a finite length curve $\cC$ as given below.

\begin{definition}
\rm({\bf Energy restricted barrier sweep coverage problem}\rm) Given a finite length curve $\cC$, an energy source $e$
and $v,t,T>0$, find the minimum number of mobile sensors with the uniform speed $v$ such that each point of $\cC$ can be
visited by at least one mobile sensor in every $t$ time period and each mobile sensor visits $e$ once in every $T$ time period.
\end{definition}

\section{Barrier sweep coverage for a finite length curve}\label{sec:single}
In this section, first we propose a solution for finding the optimal number of mobile sensors to sweep cover
a finite length curve. Later, we give an approximate solution for the energy restricted barrier sweep coverage
problem for a finite length curve.

\subsection{Optimal solution for a finite length curve}
The barrier sweep coverage problem for a finite length curve can be optimally solved using the following strategy.
Let $| \cC |$ be the length of the curve $\cC$. If $\cC$ is a closed curve, then partition $\cC$ into $\left\lceil \frac{| \cC | }{vt}\right\rceil$ equal parts of length $vt$ and deploy total $\left\lceil \frac{| \cC |}{vt}\right\rceil$ mobile sensors, one at each of the partitioning points. Then each mobile sensor starts moving at same time along $\cC$ in the same direction to ensure $t$-sweep coverage of $\cC$. If $\cC$ is an open curve, then join the end points of $\cC$ to make it close and apply the same strategy as above.

\subsection{Energy restricted barrier sweep coverage for a finite length curve}
In this section we propose an approximate solution for the energy restricted barrier sweep coverage problem.
The approximation factor of the proposed algorithm is $\frac{13}{3}$ though it is not known whether the problem is NP-hard or not.
Let $e$ be the energy source on the plane. To make the problem feasible, we assume that the distance of any point on $\cC$ from $e$ is less than $\frac{vT}{2}$. We define a tour named as $e$-tour which is denoted as $\{e,p,q,e\}$, that starts from $e$, visits $arc(pq)$ of $\cC$ continuously and then ends at $e$ such that total length of the tour is at most $vT$, where $p$ and $q$ are two points on $\cC$.
The objective of our technique is to find a tour through $e$ and $\cC$, which is concatenation of multiple number of $e$-tours.
Let $d(a,b)$ the Euclidean distance between two points $a$ and $b$ and $d_c(p,q)$ be the distance between two points $p$ and $q$ along $\cC$ in the clockwise direction, where $p$ and $q$ are two points on $\cC$. So, $d_c(p,q)$ is equal to the length of the $arc(pq)$.

\begin{figure}[h]
\psfrag{x}{$i_1$}
\psfrag{y}{$i_2$}
\psfrag{z}{$i_3$}
\psfrag{w}{$i_4$}
\psfrag{e}{$e$}
                \centering
                \includegraphics[width=0.4\textwidth]{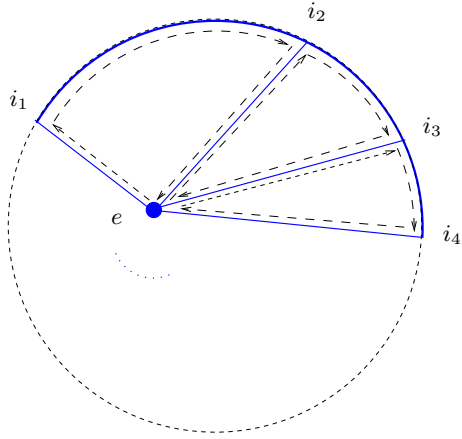}
                \caption{Showing selection of $e$-tours $\{e,i_1,i_2,e\}$, $\{e,i_2,i_3,e\}$, $\{e,i_3,i_4,e\},\cdots$}
                \label{fig:E-tour}
\end{figure}

Let us chose any point $i_1$ on a finite length `closed' curve $\cC$. Find a point $i_2$ on $\cC$ in the clockwise direction of $\cC$ (Ref. Fig. \ref{fig:E-tour}) such that $d_c(i_1,i_2)=\frac{vT}{2}-d(e,i_1)$. Here $T_1=\{e,i_1,i_2,e\}$ is a $e$-tour, since length of $T_1$ is equal to $d(e,i_1)+d_c(i_1,i_2)+d(i_2,e) \le d(e,i_1)+\frac{vT}{2}-d(e,i_1)+\frac{vT}{2}$ = $vT$, as shown in Fig. \ref{fig:E-tour}. Next $e$-tour is selected as $\{e,i_2,i_3,e\}$, where $i_3$ is a point on $\cC$ in the clockwise direction from $i_2$ such that $d_c(i_1,i_2)=\frac{vT}{2}-d(e,i_2)$.
Once the second tour is selected, we check whether the combination of the previous tour and current tour together, i.e., $\{e,i_1,i_3,e\}$ form an $e$-tour or not. If the combined tour $\{e,i_1,i_3,e\}$ is a valid $e$-tour, then the previous tour $\{e,i_1,i_2,e\}$ is updated into $\{e,i_1,i_3,e\}$ and proceed to select next $e$-tour. This updating of the previous $e$-tour with current one will continue until combined tour violates the length constraint of $e$-tour i.e., length of the combined tour is greater than $vT$. In general, after computing an $e$-tour $T_j=\{e,i_j,i_{j+1},e\}$, we select next point $i_{j+2}$ on $\cC$ such that  $d_c(i_{j+1},i_{j+2})=\frac{vT}{2}-d(e,i_{j+1})$. Then if $\{e,i_{j},i_{j+2},e\}$ this is a valid $e$-tour, we update the tour $T_j$ as $T_j=\{e,i_{j},i_{j+2},e\}$. Otherwise the tour $T_{j+1}$ is selected as $T_{j+1}=\{e,i_{j+1},i_{j+2},e\}$.
Continuing this way we decompose $\cC$ into multiple number of $e$-tours such that every point on $\cC$ is included in some $e$-tour. If $\cC$ is an `open' curve, we use the above technique to decompose $\cC$ into multiple number of $e$-tours considering one end point of $\cC$ as $i_1$ and continue till the other end point of it. Note that according to the above construction of the $e$-tour, length of the combined tour of two consecutive $e$-tours is always greater than $vT$.

Let $APPRX$ be the total tour after concatenation of the $e$-tours, $T_1,T_2,T_3,\cdots$ one after another in order of obtained by the above technique. For example $T_1=\{e,i_1,i_2,e\}$, $T_2=\{e,i_2,i_3,e\}$, $T_3=\{e,i_3,i_4,e\},\cdots$ are $e$-tours as shown in Fig. \ref{fig:E-tour}, where concatenation of those tours is $T_1 \cdot  T_2 \cdot T_3 \cdot \cdots$ =$\{e,i_1,i_2,e\} \cdot  \{e,i_2,i_3,e\} \cdot \{e,i_3,i_4,e\} \cdot \cdots$ = $\{e,i_1,i_2,e, i_2,i_3,e,i_3,i_4,e,\cdots \}$, where `$\cdot$' is denoted as the concatenation operation of the $e$-tours.
Let $ | APPRX |$ be the length of $APPRX$. Divide $APPRX$ into equal parts of length $vt$ and deploy one mobile sensor at each of the partitioning points. The mobile sensors then start moving along $APPRX$ in the same direction to ensure sweep coverage of $\cC$ for desirable long-time. The Algorithm \ref{alg:Energy}({\textsc{EnergyRestrictedBSC}}) given below for energy restricted barrier sweep coverage is applicable for closed curve. Though same Algorithm is also applicable for an open curve as explained in the previous paragraph.
\begin{algorithm}[]
\caption{\textsc{EnergyRestrictedBSC}}
\begin{algorithmic}[1]
\STATE{Choose any point $i_1$ on $\cC$.}
\STATE{$\cC'=\cC$, $n'=1$.}
\WHILE{$\cC' \ne \phi$}
\IF{$d(e,i_{n'})+ | \cC'|  \le \frac{vT}{2}$}
\STATE{$h=i_1$}
\ELSE
\STATE{Select a point $h$ on $\cC$ in a clockwise direction from $i_{n'}$ such that $d(i_{n'},h)= \frac{vT}{2}-d(e,i_{n'})$}
\ENDIF
\IF{$n' \ne 1$ and $d(e,i_{{n'}-1})+d_c(i_{{n'}-1},h)+d(e,h) \le vT$}
\STATE{$\cC'=\cC'\setminus arc(i_{n'}h)$}
\STATE{$i_{n'}=h$, $T_{{n'}-1}=\{e,i_{{n'}-1},i_{n'},e\}$}
\ELSE \label{step:violation}
\STATE{ $i_{{n'}+1}=h$, $T_{n'}$ = $\{e,i_{n'},i_{{n'}+1},e\}$}
\STATE{$\cC'=\cC'\setminus arc(i_{n'}i_{{n'}+1})$}
\STATE{$n'=n'+1$}
 \ENDIF
\STATE{$\cC'=\cC'\setminus arc(i_{n'}i_{{n'}+1})$}
\ENDWHILE
\STATE{$APPRX=T_1\cdot T_2 \cdot T_2 \cdot \cdots \cdot T_{n'}$, \\ where `$\cdot$' is denoted as concatenation operation.}
\STATE{Divide $APPRX$ into equal parts of length $vt$ and deploy one mobile sensor at each of the partitioning points.}
\STATE{All mobile sensor start moving along $APPRX$ at same time in same direction.}
\end{algorithmic}\label{alg:Energy}
\end{algorithm}
\begin{theorem}
According to the Algorithm \ref{alg:Energy}, each mobile sensor visits $e$ in every $T$ time period and each point
on $\cC$ is visited by at least one mobile sensor in every $t$ time period.
\end{theorem}
\begin{proof}
According to our proposed Algorithm \ref{alg:Energy}, the mobile sensors move along the tour $APPRX$. As the length of each $e$-tour is less than or equals to $vT$, the mobile sensors will visit $e$ after traveling at most $vT$ distance since its last visit of $e$. Therefore, each mobile sensor visits $e$ once in every $T$ time period. Again, the mobile sensors are deployed at every partitioning points of  $APPRX$ and two consecutive partitioning points are $vt$ distance apart. The relative distance between any two consecutive mobile sensors is $vt$ at any time as they are moving in same speed $v$ along the same direction. Therefore, any point on $\cC$, which is visited at time $t_0$ (say) by a mobile sensor visits again within time $t+t_0$ by the next mobile sensor following it.
\qed\end{proof}

To analyze the approximation factor with respect to the optimal solution, we consider some special points
on $\cC$ as follows. Let $i_p^1$, $i_p^2$ be two special points on the $arc(i_pi_{p+1})$ of $\cC$ such that
the $arc(i_pi_{p+1})$ is partitioned into three equal parts, i.e., the length of $arc(i_pi_p^1)$ equal to the
length of $arc(i_p^1i_p^2)$ equal to the length of $arc(i_p^2i_{p+1})$. We define a set of points, $\cI=\{i_j, i_j^1, i_j^2 | j=1~\text{to}~n' \}$.
Following two Lemmas give an upper bound of the length of the tour $APPRX$ and a lower bound of the length of the optimal tour respectively.

\begin{lemma}\label{lem:upperb}
$$|APPRX| \le \frac{1}{3} \left(2\sum_j\left(d(e,i_j)+d(e,i_j^1) \\+d(e,i_j^2)\right)+5 | \cC | \right)$$.
\end{lemma}
\begin{proof}
According to the Algorithm \ref{alg:Energy}, total length of the tour $APPRX$ is
\begin{eqnarray}\label{Eq:eq1}
|APPRX| & = & d(e,i_1)+d_c(i_1,i_2)+d(i_2,e)+d(e,i_2) \nonumber  \\ & & +d_c(i_2,i_3)  +d(i_3,e)+ \cdots + d(e,i_1)\nonumber \\
& =& | \cC | +2\sum_j d(e,i_j).
\end{eqnarray}
Now by triangle inequality,
\begin{eqnarray}
d(e,i_j) & \le & d(e,i_j^1)+d_c(i_j,i_j^1)~ {\rm and} \nonumber \\
d(e,i_j) & \le & d(e,i_j^2)+d_c(i_j,i_j^2) \nonumber
\end{eqnarray}
Therefore,
\begin{eqnarray}\label{Eq:neq1}
  | APPRX |& \le &  | \cC | + 2\sum_j\left(d(e,i_j^1)+d_c(i_j,i_j^1)\right)
\end{eqnarray}
\begin{eqnarray}\label{Eq:neq2}
  | APPRX |& \le & | \cC | + 2\sum_j\left(d(e,i_j^2)+d_c(i_j,i_j^2)\right)
\end{eqnarray}

From above the equation \ref{Eq:eq1}, \ref{Eq:neq1} and \ref{Eq:neq2}, we can write,
\begin{eqnarray}\label{Eq:neq3}
3 | APPRX | & \le & 3 | \cC |+ 2\sum_j d(e,i_j)  +  2\sum_j\left(d(e,i_j^1)+d_c(i_j,i_j^1)\right)  \nonumber  \\ & & + 2\sum_j\left(d(e,i_j^2)+d_c(i_j,i_j^2)\right)
\end{eqnarray}
As the points $i_j^1$ and $i_j^2$ divide length of the $arc(i_ji_{j+1})$ into three equal parts, therefore, \\$ \displaystyle \sum_j d_c(i_j,i_j^1)=\frac{| \cC |}{3}$ and $ \displaystyle \sum_j d_c(i_j,i_j^2)=\frac{2 | \cC |}{3}$.
Using these two results, equation \ref{Eq:neq3} can be written as:
\begin{eqnarray}
 3 | APPRX | & \le & 3 | \cC |+ 2\sum_j\left(d(e,i_j)+d(e,i_j^1) + d(e,i_j^2)\right) \nonumber  + 2\frac{ | \cC |}{3}+2\frac{2 | \cC |}{3} \nonumber \\
     | APPRX | & \le & \frac{1}{3} \left(2\sum_j\left(d(e,i_j)+d(e,i_j^1)+d(e,i_j^2)\right)\right) \nonumber  + \frac{5}{3} | \cC | . \nonumber
\end{eqnarray}
\qed\end{proof}

\begin{figure}[h]
\psfrag{x}{\hspace{-.2cm}$p$}
\psfrag{y}{$i_j$}
\psfrag{z}{$i_{j+1}$}
\psfrag{w}{\hspace{-.5cm}$i_{j-1}$}
\psfrag{p}{$q$}
\psfrag{q}{$i_{j+2}$}
\psfrag{e}{\vspace{1cm}$e$}
%\psfrag{x}{\hspace{-1cm}$i_{j-1}$}
%\psfrag{y}{$p$}
%\psfrag{z}{$i_{j}$}
%\psfrag{w}{\hspace{-1cm}$i_{j+1}$}
%\psfrag{p}{$q$}
%\psfrag{q}{$i_{j+2}$}
%\psfrag{e}{\vspace{1cm}$e$}
                \centering
                \includegraphics[width=0.4\textwidth]{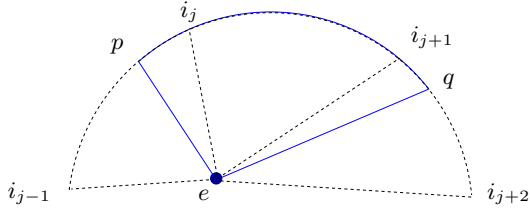}
                \caption{Showing one $e$-tour $OPT_l=\{e,p,q,e\}$ of the optimal tour $OPT$}
                \label{fig:opt-tour}
\end{figure}

\begin{lemma} \label{lem:lowerb}
Let $OPT$ be the optimal tour of the mobile sensors and $ | OPT |$ be length of $OPT$, then
$$\displaystyle | OPT | \ge  \max\left\{\frac{1}{4}\sum_j\left(d(e,i_j)+d(e,i_j^1)+d(e,i_j^2)\right), | \cC | \right\}$$.
\end{lemma}
\begin{proof}
Let $OPT_l=\{e,p,q,e\}$ be an $e$-tour of the optimal tour $OPT$ (Ref. Fig. \ref{fig:opt-tour}). We claim that at most one $arc(i_ji_{j+1})$, part of an $e$-tour computed using the Algorithm \ref{alg:Energy} for some $j$, can be completely contained in $arc(pq)$, where $arc(pq)$ is a part of $OPT_l$.
To prove this, let us assume that there are two such arcs $arc(i_{j-1}i_j)$ and $arc(i_ji_{j+1})$ are completely contained in $arc(p,q)$.
As length of combined tour of  two consecutive $e$-tours is always greater than $vT$ according to step \ref{step:violation} of the Algorithm \ref{alg:Energy}, therefore,
\begin{equation}\label{eq:vt}
\left( d(e,i_{j-1})+d_c(i_{j-1},i_{j+1})+d(e,i_{j+1})\right) > vT
\end{equation}
Now,
\begin{eqnarray}
| OPT_l | &=& d(e,p)+d_c(p,i_{j-1})+d_c(i_{j-1},i_j)\nonumber \\ & &  +d_c(i_j,i_{j+1})+d_c(i_{j+1},q)+d(e,q)\nonumber \\
& \ge& d(e,i_{j-1})+d_c(i_{j-1},i_{j+1})+d(e,i_{j+1}) \nonumber\\
& > & vT ~ \rm{(from~ equation~ \ref{eq:vt})} \nonumber
\end{eqnarray}
This contradicts the fact that $OPT_l$ is an $e$-tour. Therefore, at most one $arc(i_ji_{j+1})$ for some $j$, can be completely contained in $arc(pq)$. Hence, at most one complete $arc(i_ji_{j+1})$ and two $arc(i_{j-1}i_{j})$ and $arc(i_{j+1}i_{j+2})$ partially contained in $arc(pq)$ as shown in Fig. \ref{fig:opt-tour}.
Therefore, maximum eight points from the set $\cI$ may belong to the $arc(pq)$, since there are four points $i_j$, $i_j^1$, $i_{j}^2$, $i_{j+1}$ for the $arc(i_ji_{j+1})$ and at most four special points, $i_{j-1}^1$, $i_{j-1}^2$ and $i_{j+1}^1$, $i_{j+1}^2$ for the $arc(i_{j-1}i_{j})$ and $arc(i_{j+1}i_{j+2})$ respectively.

Now, for any point $x$ on $arc(pq)$ of $OPT_l$ implies  $ | OPT_l |  \ge 2d(e,x)$. As there are at most eight points of $\cI$ in $arc(pq)$, which implies
\begin{eqnarray}
\displaystyle  | OPT_l |  & \ge & \frac{2\sum_{x ~\in~ \cI \cap arc(pq)} d(e,x)}{8}
\end{eqnarray}
Since all the points in $\cI$ are on $arc(pq)$ for some $OPT_l$, where $OPT_l$ is a part of $OPT$, therefore,\\
$| OPT |  \ge \sum _l | OPT_l | \ge \frac{2\sum_j\left(d(e,i_j)+d(e,i_j^1)+d(e,i_j^2)\right)}{8}$.
Also, $ | OPT |  \ge  | \cC | $.
Therefore, \\$ | OPT | \ge \max\left\{\frac{1}{4}\sum_j\left(d(e,i_j)+d(e,i_j^1) \\ +d(e,i_j^2)\right), | \cC | \right\}$.
\qed\end{proof}

\begin{theorem}
The approximation factor of the Algorithm \ref{alg:Energy} for energy restricted barrier sweep coverage problem is $\displaystyle \frac{13}{3}$.
\end{theorem}

\begin{proof}
Let $N$ be the number of mobile sensors needed in our solution and $N_{opt}$ be the number of mobile sensors in the optimal solution.\\ Then $\displaystyle N =\left\lceil\frac{ | APPRX | }{vt}\right\rceil$ and $\displaystyle N_{opt}\ge \frac{ | OPT | }{vt}$.\\
From Lemma \ref{lem:upperb} and Lemma \ref{lem:lowerb}, we have $\displaystyle\frac{N}{N_{opt}} \le \frac{ | APPRX | }{ |  OPT |  } \le \frac{8}{3}+\frac{5}{3} =\frac{13}{3}$.\\ Hence the approximation factor of our proposed Algorithm  \ref{alg:Energy} is $\displaystyle \frac{13}{3}$.
\qed\end{proof}

\section{Barrier sweep coverage problem for multiple finite length curves}\label{sec:multiple}
%\begin{theorem}\label{th:NPC}
%Barrier sweep coverage problem for multiple finite length curves (BSCMC) is  NP-hard and cannot be approximated within a
%factor of $2$ unless P=NP.
%\end{theorem}
%\begin{proof}
%Finding minimum number of mobile sensors with uniform velocity to guarantee sweep coverage for a set of points
%in 2D plane is NP-hard and it cannot be approximated within a factor of 2 unless P=NP, as proved
%in the paper \cite{Li11} by Li et al. The point sweep coverage problem as proposed in \cite{Li11} is a special
%case of BSCMC when all curves are points. Therefore, BSCMC is NP-hard and cannot be approximated within a factor of 2 unless P=NP.
%\qed\end{proof}

Finding minimum number of mobile sensors with uniform velocity to guarantee sweep coverage for a set of points
in 2D plane is NP-hard and it cannot be approximated within a factor of 2 unless P=NP, as proved
in the paper \cite{Li11} by Li et al. The point sweep coverage problem as proposed in \cite{Li11} is a special
case of BSCMC when all curves are points. Therefore, BSCMC is NP-hard and cannot be approximated within a factor of 2 unless P=NP.

\subsection{2-approximation Solution for a Special Case}\label{sec:FLSC}
In this section we propose a solution for BSCMC for a special case where each mobile sensor must visits all the points of each curve. We propose an algorithm where each curve is a line segment. The same idea works for finite length curves explained in the last paragraph of section \ref{sec:multiple}.
\begin{figure}[h]
\psfrag{a}{$a$}
\psfrag{b}{$b$}
\psfrag{l}{$l$}
                \centering
                \includegraphics[width=0.4\textwidth]{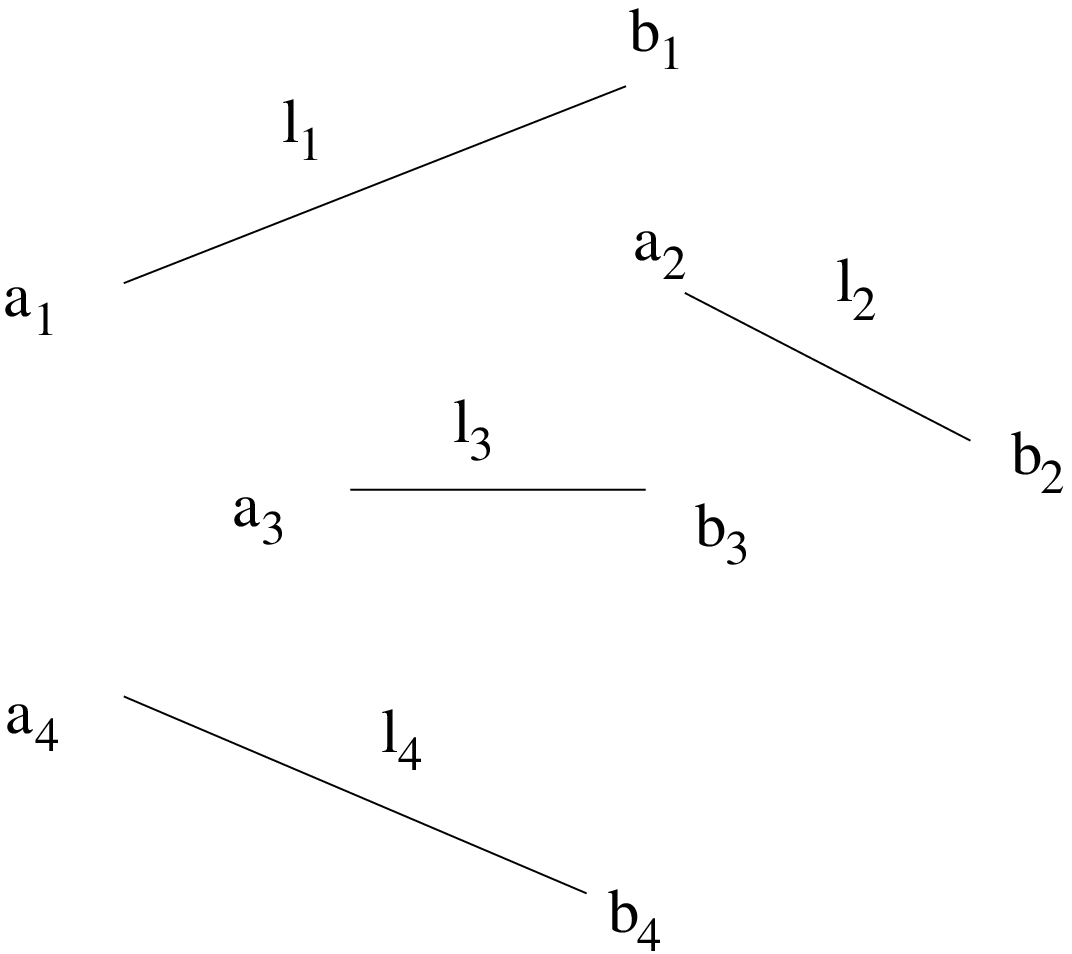}
                \caption{Set of line segments $\cL$}
                \label{fig:setOfLine}
\end{figure}
Let $\cL=\{l_1,l_2,\cdots,l_n\}$ be a set of line segments on a 2D plane. Let $S$ be the set of shortest distance line $s_{ij}$ between every pair of line segments ($l_i,l_j$) for $i\ne j$. We define a complete weighted graph $G=(V,E)$, where $V=\{v_1,v_2,\cdots,v_n\}$ is the set of vertices. The vertex $v_i$ represents line segment $l_i$ for $i= 1$ to $n$. $E=V\times V$ is the set of edges, where the edge $(v_i,v_j)$ represents $s_{ij} \in S$ and edge weight $w(v_i,v_j)=$ length of $s_{ij}$. Let $T$ be a minimum spanning tree (MST) of $G$.  $T$ can be represented as $T_\cL$, where $T_\cL=\cL\cup\{s_{ij} : (v_i,v_j) \in T \}$. An illustration is shown from Fig. \ref{fig:setOfLine} to Fig. \ref{fig:TL}, where a set of line segments is shown in Fig. \ref{fig:setOfLine}, corresponding complete graph $G$ is shown in Fig. \ref{fig:CompleteG}, an MST $T$ of $G$ is shown in Fig. \ref{fig:MST} and the representation $T_\cL$ of $T$ is shown in Fig. \ref{fig:TL}.

We construct a graph $G_\cL$ from $T_\cL$ by introducing vertices at the end points of each line segment in $T_\cL$, which may split $l_i$'s into several smaller line segments. According to Fig. \ref{fig:GL}, vertices of $G_\cL$ are \{$a_1, p, b_1$, $a_2,q,b_2$, $a_3,b_3$, $a_4,r,b_4$\}. The vertex $p$ splits line segment ($a_1,b_1$) into two smaller line segments ($a_1,p$) and ($p,b_1$). Similarly, vertices $q$ and $r$ split ($a_2,b_2$) and ($a_4,b_4$) into ($a_2,q$), ($q,b_2$) and ($a_4,r$), ($r,b_4$) respectively, whereas the line segment ($a_3,b_3$) remains same.
Each of these line segments and the lines corresponding to the edges of $T$ together are the edges of the graph $G_\cL$. According to Fig. \ref{fig:GL}, edges of $G_\cL$ are \{($a_1,p$), ($p,b_1$), ($a_2,q$), ($q,b_2$), ($a_3,b_3$), ($a_4,r$), ($r,b_4$), ($p,a_3$), ($a_3, r$), ($b_3,q$)\}.

\begin{minipage}{0.45\linewidth}
          \begin{figure}[H]
\psfrag{v}{$v$}
\psfrag{1}{$_1$}
\psfrag{2}{$_2$}
\psfrag{3}{$_3$}
\psfrag{4}{$_4$}
                \centering
              \includegraphics[width=.73\linewidth]{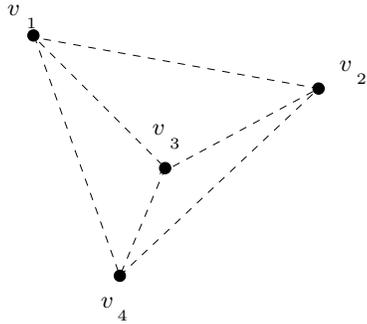}
              \caption{Complete graph $G$}\label{fig:CompleteG}
          \end{figure}
\end{minipage}
\hspace*{.85cm}
\begin{minipage}{0.45\linewidth}
          \begin{figure}[H]
\psfrag{v}{$v$}
\psfrag{1}{$_1$}
\psfrag{2}{$_2$}
\psfrag{3}{$_3$}
\psfrag{4}{$_4$}
                \centering
              \includegraphics[width=.7\linewidth]{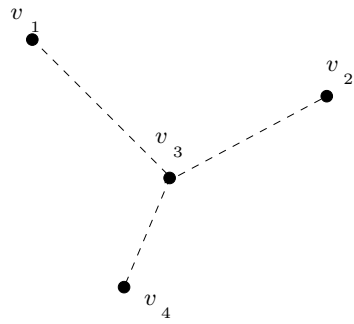}
              \caption{MST $T$ of $G$}\label{fig:MST}
          \end{figure}
\end{minipage}

\begin{minipage}{0.45\linewidth}
          \begin{figure}[H]
\psfrag{a}{$a$}
\psfrag{b}{$b$}
\psfrag{p}{$p$}
\psfrag{q}{$q$}
\psfrag{r}{$r$}
                \centering
              \includegraphics[width=\linewidth]{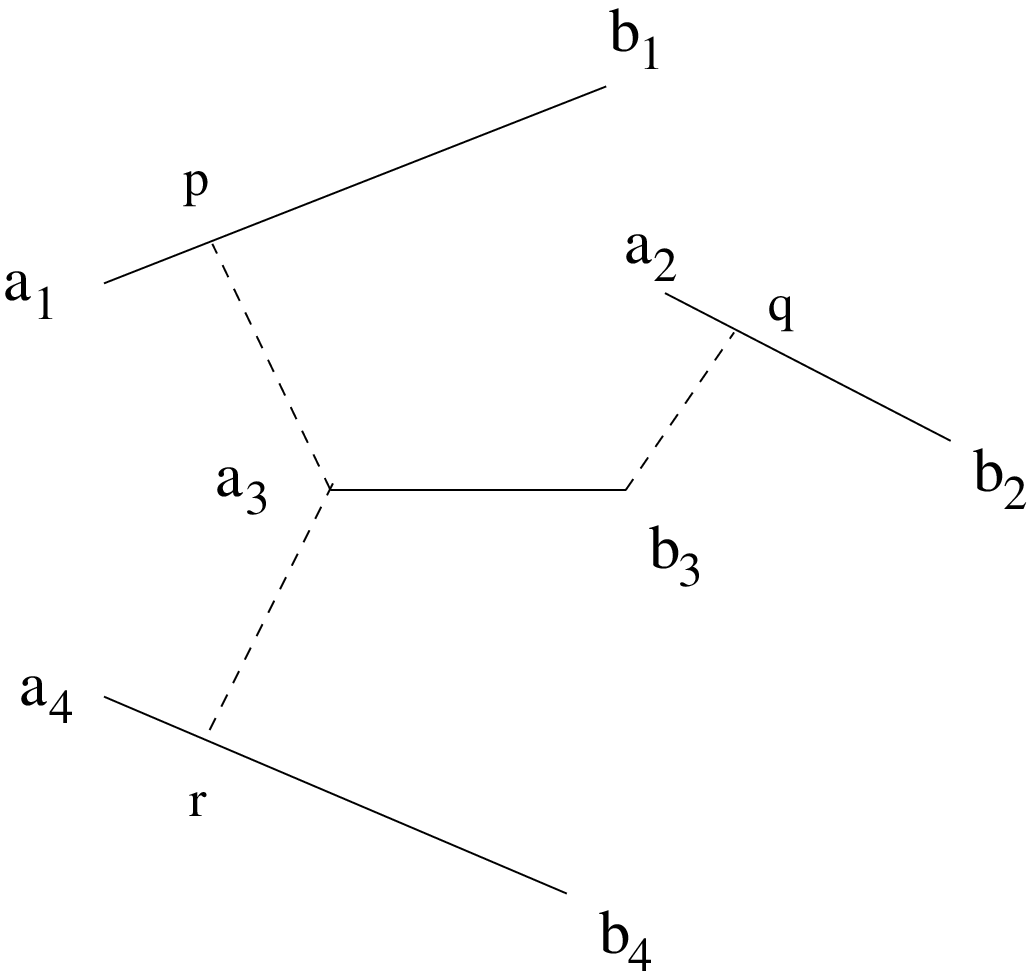}
              \caption{$T_{\cL}$}\label{fig:TL}
          \end{figure}
\end{minipage}
\hspace*{.85cm}
\begin{minipage}{0.45\linewidth}
          \begin{figure}[H]
          \psfrag{a}{$a$}
\psfrag{b}{$b$}
\psfrag{p}{$p$}
\psfrag{q}{$q$}
\psfrag{r}{$r$}
                \centering
              \includegraphics[width=\linewidth]{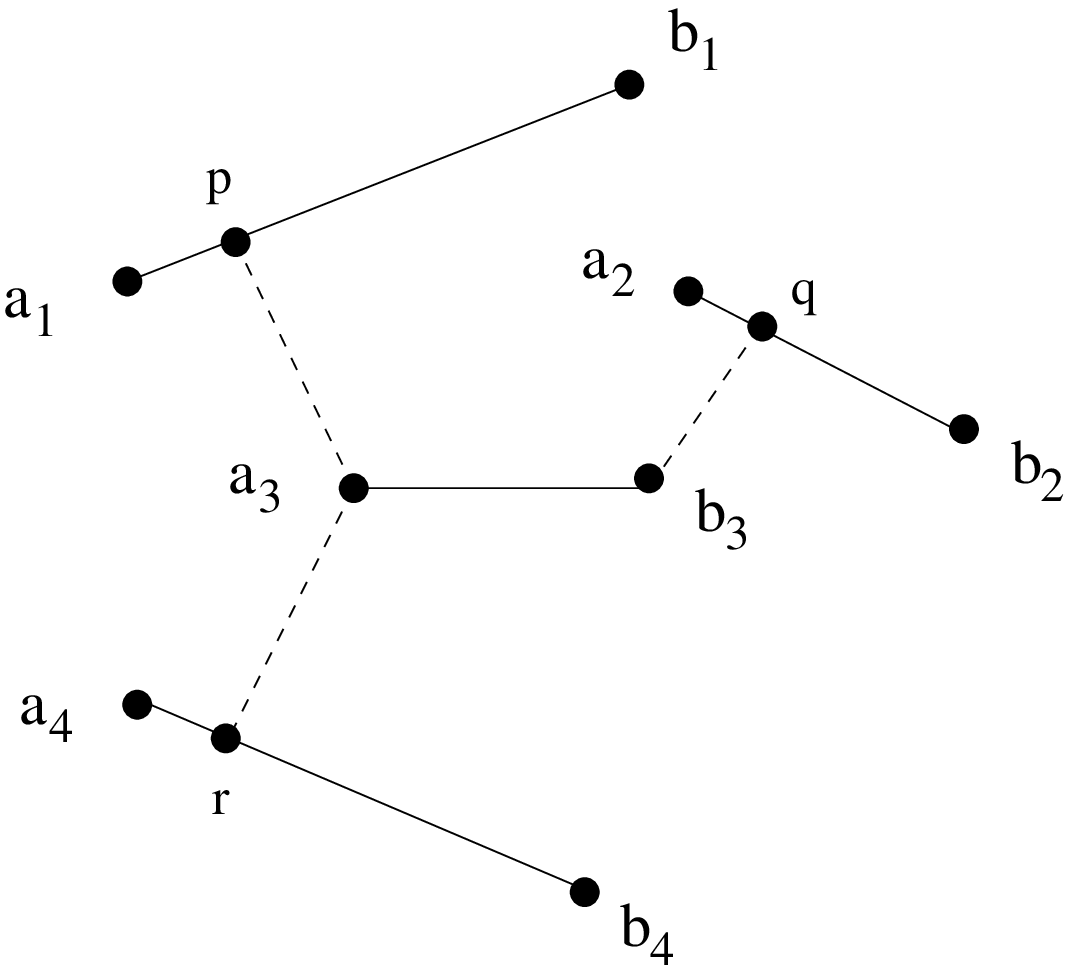}
              \caption{$G_{\cL}$}\label{fig:GL}
          \end{figure}
\end{minipage}

%
%          \begin{figure}[ht]
%\psfrag{v}{$v$}
%\psfrag{1}{$_1$}
%\psfrag{2}{$_2$}
%\psfrag{3}{$_3$}
%\psfrag{4}{$_4$}
%                \centering
%              \includegraphics[width=0.25\textwidth]{fig_2.eps}
%              \caption{Complete graph $G$}\label{fig:CompleteG}
%          \end{figure}
%
%          \begin{figure}[ht]
%\psfrag{v}{$v$}
%\psfrag{1}{$_1$}
%\psfrag{2}{$_2$}
%\psfrag{3}{$_3$}
%\psfrag{4}{$_4$}
%                \centering
%              \includegraphics[width=0.25\textwidth]{fig_3.eps}
%              \caption{MST $T$ of $G$}\label{fig:MST}
%          \end{figure}
%
%          \begin{figure}[ht]
%\psfrag{a}{$a$}
%\psfrag{b}{$b$}
%\psfrag{p}{$p$}
%\psfrag{q}{$q$}
%\psfrag{r}{$r$}
%                \centering
%              \includegraphics[width=0.25\textwidth]{fig_4.eps}
%              \caption{$T_{\cL}$}\label{fig:TL}
%          \end{figure}
%
%          \begin{figure}[ht]
%          \psfrag{a}{$a$}
%\psfrag{b}{$b$}
%\psfrag{p}{$p$}
%\psfrag{q}{$q$}
%\psfrag{r}{$r$}
%                \centering
%              \includegraphics[width=0.25\textwidth]{fig_5.eps}
%              \caption{$G_{\cL}$}\label{fig:GL}
%          \end{figure}

The graph $G_{\cL}$ is a tree and the sum of the edge weights of the graph is $| G_{\cL}| $ = $| T | +\sum_{i=1}^{n}{l_i}$,
where $ | T |$ is the sum of the edge weights of $T$. Following Algorithm \ref{alg:pint1} ({\textsc BarrierSweepCoverage}) computes a tour
on $G_\cL$ and finds number of mobile sensors and their movement paths.

\begin{lemma}
According to the Algorithm \ref{alg:pint1} each point on $l_i$ can be visited by at least
one mobile sensor in every time period $t$ for $i=1,2,\cdots n$.
\end{lemma}
\begin{proof}
Since the mobile sensors are moving along the Eulerian tour $\cE$, each edges of $G_\cL$ are visited.
Let us consider any point $p$ on a line segment $l_i$ and let $t'$ be the
time when a mobile sensor visited $p$ last time.
Now we have to prove that the point $p$ must be visited by at least one mobile sensor in $t'+t$ time.
According the deployment strategy of mobile sensors any two consecutive mobile sensors are within
the distance of $vt$ at any time. So, when a mobile sensor visited $p$ at $t'$ another mobile sensor
is on the way to $p$ and within the distance of $vt$ along $\cE$. Hence $p$ will be again visited by
another mobile sensor within next $t$ time.
\qed\end{proof}
\begin{algorithm}[]
\caption{\textsc{BarrierSweepCoverage}}
\begin{algorithmic}[1]
\STATE{Construct complete weighted graph $G$ from the given set of line segments $\cL$.}
\STATE{Find an MST $T$ of $G$.}
\STATE{Construct $G_{\cL}$.}
\STATE{Find Eulerian graph after doubling each edge of $G_\cL$.}
\STATE{Find an Eulerian tour $\cE$ on the Eulerian graph. Let $|\cE|$ be the length of $\cE$.}
%\STATE{Find a tour $L$ from the euler tour $\cE$ using Algorithm \ref{alg:tour}.}
\STATE{Partition $\cE$ into $\left\lceil\frac{|\cE|}{vt}\right\rceil$ parts and deploy $\left\lceil{\frac{|\cE|}{vt}}\right\rceil$
mobile sensors at all partition points, one for each.}
\STATE{Each mobile sensor then starts moving at the same time along $\cE$ in same direction.}
\end{algorithmic}\label{alg:pint1}
\end{algorithm}

\begin{lemma}\label{lem:opt}
If $L_{opt}$ is the length of the optimal TSP tour for visiting all points of every line segment in $\cL$,
then $| T |+\sum_{i=1}^{n}{l_i}\le L_{opt}$.
\end{lemma}
\begin{proof}
The optimal TSP tour $L_{opt}$ contains two types of movement paths; movement paths along the line segments of $\cL$
and movement paths between the line segments.
Let total length of the movement paths along the line segments be $L_{along}$.  Since $L_{opt}$ is the optimal tour
for visiting all points of each line segment $l_i \in \cL$, therefore,
\begin{equation}\label{Eq:Lalong}
   L_{along}\ge \sum_{i=1}^{n}{l_i}
\end{equation}
Let $L_G$ be the optimal TSP tour on $G$. Then $ | T | \le L_G $.
Let total length of the movement paths between the line segments be $L_{between}$. Since $L_G$ is the optimal TSP
tour on $G$ and the weights of all edges of $G$ are taken to be the shortest distance between respective line segments,
therefore, \begin{equation}\label{Eq:Lbetween}
   L_{between}\ge L_G
\end{equation}
Now, from equation \ref{Eq:Lalong} and  \ref{Eq:Lbetween}, $L_{G} +\sum_{i=1}^{n}{l_i} \le L_{opt}$.
Hence $| T | + \sum_{i=1}^{n}{l_i}\le L_{opt}$.
\qed\end{proof}
\begin{theorem}
The approximation factor of the Algorithm \ref{alg:pint1} is 2.
\end{theorem}
\begin{proof}
The total edge weights of $G_\cL$ is $| T |+\sum_{i=1}^{n}{l_i}$.
Now, $|\cE| = 2(| T |+\sum_{i=1}^{n}{l_i})$, since Eulerian tour $\cE$ found by the
Algorithm \ref{alg:pint1} after doubling each edges of $G_\cL$.
By Lemma \ref{lem:opt}, $ |\cE|\le 2L_{opt}$. Let $N_{opt}$ be the number of mobile
sensors required for optimal solution. Then $N_{opt} \times vt \ge L_{opt}$, {\it i.e.},
$N_{opt} \ge \left\lceil \frac{L_{opt}}{vt}\right\rceil$. The number of mobile sensors
calculated by the Algorithm \ref{alg:pint1} is $\left\lceil\frac{|\cE|}{vt}\right\rceil$
(=$N$, say). Therefore, the approximation factor of the Algorithm \ref{alg:pint1} is equal to
$\frac{N}{N_{opt}} \le \left\lceil\frac{2L_{opt}}{vt}\right\rceil\Big/\left\lceil\frac{L_{opt}}{vt}\right\rceil\le 2$.
\qed\end{proof}

\subsection{Solution for BSCMC}
In this section, we propose an algorithm for the BSCMC problem. The description of the algorithm is given below.

For $k=1$ to $n$ following calculations are performed.
Compute a minimum spanning forest $F_k$ of $G$ with $k$ components.
Let $C^1$, $C^2$, $\cdots$, $C^k$ be the connected components of $F_k$.
Let $C^i_{\cL}$ be the representation of $C^i$ on the set of line segments $\cL$ for $i=1$ to $k$,
where $C^i_{\cL}= \{l_j| v_j \in C^i\} \cup\{s_{ij}| (v_i,v_j) \in C^i\}$.

Construct graph $G^i_{\cL}$ from $C^i_{\cL}$ in the same way the $G_\cL$ is constructed from $T_\cL$ in section \ref{sec:FLSC}.
Clearly, $\sum _{i=1}^k {| G^i_{\cL}|}= \sum_{j=1}^n l_j + | F_k |$.

Find Eulerian tours $\cE^1_{\cL}$, $\cE^2_{\cL}$, $\cdots$, $\cE^k_{\cL}$ after doubling the edges of $G^1_{\cL}$,
$G^2_{\cL}$, $\cdots$, $G^k_{\cL}$ respectively. Partition each $\cE^k_{\cL}$ into $\left\lceil\frac{|\cE^k_{\cL}|}{vt}\right\rceil$ parts of length $vt$. Let $N_k$ be the total number of partitioning points. Choose the minimum over all $N_k$'s as the number of mobile sensors. Deploy the number of mobile sensors, one at each of the partitioning points. Then all mobile sensors start their movement at the same time along
their respective tours in the same direction.
\begin{theorem}
The approximation factor of the proposed solution for BSCMC is 5.
\end{theorem}
\begin{proof}
Let $opt$ be the number of mobile sensor required for the optimal solution of BSCMC. Let $opt'$ be the minimum number of mobile sensor which can
guarantee $t$-sweep coverage of all the vertices of $G$. Then $opt \ge opt'$.

Since, all the points of each line segments are visited by the mobile sensors in any time period $t$ then $opt \ge \frac{\sum_{i=1}^{n}{l_i}}{vt}$.
Let us consider the movement paths of $opt'$ number of mobile sensors to sweep cover the vertices of $G$ in any time interval $[t_0,t+t_0]$. Let $Min\_path$ be the total sum of the lengths of the paths.
Then $Min\_path \le vt\times opt'$. Again these $opt'$  movement paths form a spanning forest with $opt'$ number of connected components of $G$.
Therefore, we have $ |  F_{opt'} |   \le Min\_path$ as $F_{opt'}$ is the minimum spanning forest of $G$ with $opt'$ components.
Let us consider the iteration of our solution for $k=opt'$.
Total number of mobile sensors $N$ in this iteration is given below.\\
$N=\sum _{i=1} ^k \left \lceil \frac{|\cE^i_{\cL}|}{vt}\right \rceil \le \sum _{i=1} ^k \frac{|\cE^i_{\cL}|}{vt} + k =
2\sum _{i=1}^k \frac{| G^i_{\cL}|}{vt} + k = 2\frac{\sum _{i=1} ^n l_i}{vt}+2\frac{| F_k|}{vt} + k
\le 2opt+ 2 \frac{Min\_path}{vt}+k \le 2 opt+ 2 opt' + opt' \le 5opt$.

Therefore, the approximation factor of our proposed solution is 5.
\qed\end{proof}

The two algorithms proposed in this section also works for a set of finite length curves as explained below.
Let $\cX$ be a set of finite length curves and $S$ be the set of shortest distance line between
every pair of the curves. The complete weighted graph $G$ can be constructed considering each
curve as a vertex and the distance between every pair of curves as an edge weight corresponding
to the edge. For the MST or any subtree of $G$, Eulerian tour on the tree or the subtree can be form in the same way by introducing vertices vertices at the end points of each curve and the joining line segments and doubling the edges. The number of mobile sensors also can be deployed after partitioning the tours at each of the partitioning points. The mobile sensors follow same movement strategy
as earlier to guarantee sweep coverage of the multiple finite length curves.

\section{Data Gathering by Data Mules}\label{sec:MDMDG}
In this section we consider a data gathering problem by a set of data mules
\cite{Anastasi08,CelikM10,LevinES14,LevinSS10,ShahRJB03}, which we have formulated
as a variation of barrier sweep coverage problem. A set of mobile sensors are
moving along finite length straight lines on a plane for monitoring or
sampling data around it. Movement of the mobile sensors are arbitrary along their
respective paths i.e., a mobile sensor moves in any direction along the straight line
with its arbitrary speed. A set of data mules are moving with uniform speed $v$ in the same
plane for collecting data from the mobile sensors. A data mule can collect data from a mobile
sensor whenever it meets the mobile sensor on its path. We assume data transfer can be done
instantaneously whenever a data mule meets a mobile sensor.
The definition of the problem is given below.

\begin{definition} \rm({\it Minimum number of data mule for data gathering \rm(MDMDG\rm)}\rm)
A set of mobile sensors are moving arbitrarily on a plane along line segments. Find minimum
number of data mules, which are moving with uniform speed $v$, such that data can be collected
from each of the mobile sensors at least once in every $t$ time period.
\end{definition}

The points sweep coverage problem \cite{Li11} is a special instance of the MDMDG when two
end points of every line segment are same, {\it i.e.} the mobile sensors are behaving like
a static sensor. Therefore, MDMDG problem is NP-hard and cannot be approximated within a factor of 2 unless P=NP.

The following Lemma \ref{lem:MDMDG1} shows that to visit the mobile sensors, each point of
all paths must be visited by the data mules.

\begin{lemma}\label{lem:MDMDG1}
To solve the  MDMDG problem, each and every point of all line segments must be visited by the set
of data mules.
\end{lemma}

\begin{proof}
We will prove the Lemma by the method of contradiction.  Let $l$ be a line segment for which all points
of $l$ are not visited by the data mules. Therefore, there exist one point $p$ on $l$ such that $p$
is not visited by any data mule. One mobile node can stop its movement for sometime and which can be
allowed for the arbitrary nature of their movements. Now, if the mobile sensor on $l$ remains static
at $p$ for more than $t$ time then the mobile sensor is not visited by any data mule and which
contradict the condition of  $t$ sweep coverage. Hence, each and every point of all line segments
must be visited by the set of data mules.
\qed\end{proof}

Now, we find the minimum path traveled by a single data mule to visit all the mobile sensors.
According to the problem definition and by Lemma \ref{lem:MDMDG1}, within any time interval $[t_0,t_0+t]$,
all points of each line segment are visited by the data mules. Therefore, the total length
of the tour traversed by all the data mules within the time interval is greater or equals to the optimal
tour traversed by a single data mule in order to visit all mobile sensors.
The following Lemma gives the nature of the optimal tour traversed by a single data mule for visiting
all the mobile sensors.

\begin{lemma}\label{lem:MDMDG2}
It may not be possible to visit a mobile sensor by a data mule unless it visits the whole
line segment, which is movement path of the mobile sensor, from one end to the other end continuously.
\end{lemma}

\begin{proof}
If a line segment is not visited continuously then in the following scenario the data mule cannot
visit a mobile sensor. Let $l'$ and $l''$ be the two parts of a line segment $l$. The data mule
continuously visits $l'$ and during that time period the mobile sensor remains on $l''$. After sometime,
when the mobile sensor remains on $l'$, the data mule visits $l''$ continuously. So, in this scenario
the data mule  cannot visit mobile sensor.
\qed\end{proof}

Let $L_{opt}$ be the optimal tour for visiting all mobile sensors by one data mule. Now, the optimal tour
$L_{opt}$ contains two types of movement paths: the movement paths along the line segments and the paths between
pair of line segments. According to Lemma \ref{lem:MDMDG2}, the paths between pair of line segments are the lines which
connect end points of the pair of line segments.

We construct  euclidian complete graph  $G_{2n}$ with $2n$ vertices $a_i,b_i$, $i=1,2, \cdots n$,
where $a_i$ and $b_i$ are the two end points of the line $l_i$. The edge set $E(G_{2n})$ of $G_{2n}$ is given by
$E(G_{2n})= \{l_i :  i=1,2,\cdots, n\} \cup \{(a_i,a_j):  i \ne j\} \cup  \{(b_i,b_j):  i \ne j\}
\cup \{(a_i,b_j):  i \ne j\} \cup \{(b_i,a_j): i \ne j\}$. The weight of each of the edges is equal to
the euclidian distance between the two vertices.

Let $T_{2n}$ be a MST of $G_{2n}$ containing all edges $l_i \in E(T_{2n}), i=1,2,\cdots, n$.
We compute $T_{2n}$ using Kruskal's algorithm after including  all edges $l_i, i=1,2,\cdots,n$ in the
initial edge set of $T_{2n}$. Until the spanning tree is formed we apply Kruskal's algorithm on the
remaining edges, $E(G_{2n})\backslash \{l_i :  i=1,2,\cdots, n\}$ of $G_{2n}$. An Eulerian graph  is
formed from $T_{2n}$ as described in Christofides algorithm \cite{christofides76}.
We compute an Eulerian tour $\cE_{2n}$ from the above Eulerian graph.

We cannot directly apply the movement strategy of the mobile sensors, which is used in Algorithm \ref{alg:Energy} and
Algorithm  \ref{alg:pint1} to give the movement strategy of the data mules to solve the MDMDG problem.
If there exist a line segment with length greater than $vt$ and the mobile sensor moves with a speed $v$ along the line
in the same direction as the data mule moves then it may not be possible by the data mule to meet the mobile sensor within
time $t$. Hence, we apply a new strategy as explained below in order solve the MDMDG problem.

Partition the tour $\cE_{2n}$ into equal parts of length $vt$ and consider two sets of data mules $DM_1$ and $DM_2$,
each of which contains $\lceil\frac{\cE_{2n}}{vt}\rceil$ number of data mules. Deploy two data mules at each of the
partitioning points one from each set. Then each data mule from the set $DM_1$ moves in a direction say, clockwise
direction, whereas other set of data mules $DM_2$ move in the counter clockwise direction. All data mules,
irrespective of the sets  start their movement at same time. Based on the above discussions, we propose following
Algorithm \ref{alg:MDMDG} ({\textsc{MDMDG}}).

\begin{algorithm}
\caption{\textsc{MDMDG}}
\begin{algorithmic}[1]
%\STATE{$G_{2n}$ be the complete weighted graph corresponding to $2n$ points $a_1, b_1,a_2,b_2,\cdots,a_n,b_n$.}
\STATE{Use Kruskal algorithm to find an MST  $T_{2n}$ of $G_{2n}$ with the initial set of edges containing all edges $l_i$ for $i=1,2,\cdots,n$.}
\STATE{Construct an Eulerian graph  from $T_{2n}$ using Christofides algorithm \cite{christofides76}.}
\STATE{Find an Eulerian tour $\cE_{2n}$ from the  Eulerian graph. Let $|\cE_{2n}|$ be the length of $\cE_{2n}$.}
\STATE{Partition $\cE_{2n}$ into $\left\lceil\frac{|\cE_{2n}|}{vt}\right\rceil$ parts of length  $vt$. Deploy two data
mules, one from $DM_1$ and other from $DM_2$ at each of the partitioning points.}
\STATE{All data mules start moving at the same time along $\cE_{2n}$ such that the data mules from $DM_1$ move in clockwise direction and the data mules from $DM_2$ move in anticlockwise direction.}
\end{algorithmic}\label{alg:MDMDG}
\end{algorithm}
\subsection{Analysis}
\begin{theorem}[Correctness]
According to the Algorithm \ref{alg:MDMDG}, each mobile sensor is visited by a data mule at least once in every $t$ time period.
\end{theorem}
\begin{proof}
Let the position of a mobile sensor be $p$ when it last visited by a data mule at time $t_0$. According the deployment strategy of data mules,
two data mules visit $p$ again within time $t_0+t$, from two different directions. The statement of the theorem follows if the mobile sensor remains static at $p$ till $t_0+t$. Now we consider the case when the mobile sensor moves clockwise direction from $p$ after time $t_0$. In this case there exist a data mule, which is moving in counter clockwise direction visits it within time $t_0+t$. Similarly, if the mobile sensor moves counter clockwise direction from $p$ after time $t_0$, then there exist a data mule, which is moving in clockwise direction visits it within time $t_0+t$.
Hence, irrespective of the nature of movement, a mobile sensor is visited by a data mule at least once in every $t$ time period.
\qed\end{proof}
\begin{theorem}
The approximation factor of the Algorithm \ref{alg:MDMDG} is 3.
\end{theorem}
\begin{proof}
According to the Christofides algorithm \cite{christofides76} we can write $|\cE_{2n}| \le \frac{3}{2}L_{opt}$.
Let $N_{opt}$ be the number of data mules required for optimal solution.\\ Then $N_{opt} \times vt \ge L_{opt}$, {\it i.e.}, $N_{opt} \ge \left\lceil \frac{L_{opt}}{vt}\right\rceil$. The number of data mule calculated by the Algorithm \ref{alg:MDMDG} is $\left\lceil\frac{2|\cE_{2n}|}{vt}\right\rceil$ (=$N$, say). Therefore, the approximation factor for the Algorithm \ref{alg:MDMDG} is equal to $\frac{N}{N_{opt}} \le \left\lceil\frac{2 \times\frac{3}{2}L_{opt}}{vt}\right\rceil\Big/\left\lceil\frac{L_{opt}}{vt}\right\rceil\le 3$.
\qed\end{proof}
\section{Simulation Results}
To the best of our knowledge there is no related work on barrier sweep coverage problem in literature.
We compare performance of our proposed Algorithm \ref{alg:pint1} and the Algorithm for BSCMC through simulation.
We implement both of the algorithms in C++ language. A set of line segments are randomly generated inside a square region of side 200 meter.
The length of each line segment is randomly chosen within 5 meter. The velocity of each mobile sensor is taken as 1 meter per second.

\begin{table}[]
\centering
{\small
\begin{tabular}{|c|c|c|}
\hline
No. of line segments & \multicolumn{2}{c|}{Number of Mobile Sensors}\\
\cline{2-3}
& \multicolumn{1}{c|}{Algorithm for BSCMC} &\multicolumn{1}{c|}{Algorithm \ref{alg:pint1}}\\
\hline
5&5&12\\
15&14&28\\
25&21&35\\
35&26&43\\
45&35&54\\
55&42&64\\
65&51&72\\
75&57&77\\
85&66&90\\
95&75&97\\
105&76&100\\
115&80&104\\
125&83&107\\
135&86&110\\
\hline
\end{tabular}
\caption{Average number of mobile sensors to achieve sweep coverage varying with number of line segments for fixed sweep period}\label{tab2}
}
\end{table}

\begin{figure}[]
  \centering
    \includegraphics[width=0.65\textwidth]{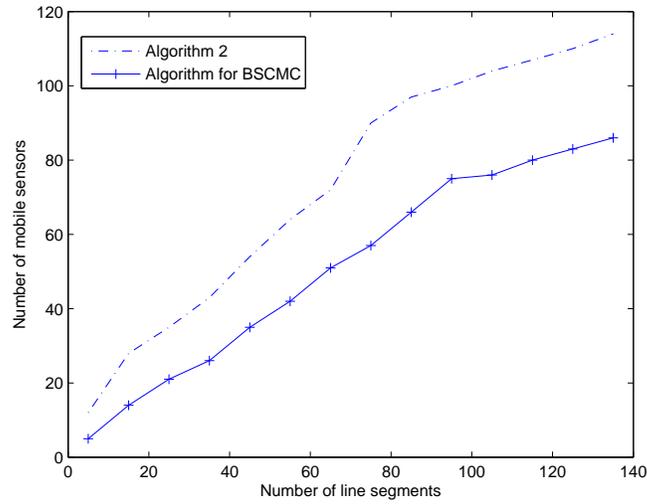}
    \caption{Comparison with respect to  number of mobile sensors varying with number of line segments}\label{fig1:vsn}
\end{figure}
\begin{table}[]
\centering
{\small
\begin{tabular}{|c|c|c|}
\hline
Sweep period & \multicolumn{2}{c|}{Number of Mobile Sensors}\\
\cline{2-3}
& \multicolumn{1}{c|}{Algorithm for BSCMC} &\multicolumn{1}{c|}{Algorithm \ref{alg:pint1}}\\
\hline
50&75&197\\
60&67&179\\
70&62&135\\
80&58&119\\
90&56&107\\
100&54&97\\
110&52&86\\
120&51&77\\
130&51&71\\
140&50&66\\
150&50&60\\
\hline
\end{tabular}
\caption{Comparison of the number of mobile sensors to achieve sweep coverage varying sweep period for fixed number of line segments}\label{tab1}
}
\end{table}

Table \ref{tab2} shows comparison of average number of mobile sensors to achieve sweep coverage for both of the algorithms varying with number of line segments. The average number of mobile sensor is calculated for 100 several executions of the algorithms with fixed sweep period
50 second. A graphical representation of the Table \ref{tab2} is illustrated in Fig. \ref{fig1:vsn}.
Table \ref{tab2} and Fig. \ref{fig1:vsn} show that with increasing number of line segments, the
Algorithm BSCMC performs better than the Algorithm \ref{alg:pint1} with respect to average number of mobile sensors.

\begin{figure}[ht]
  \centering
    \includegraphics[width=0.65\textwidth]{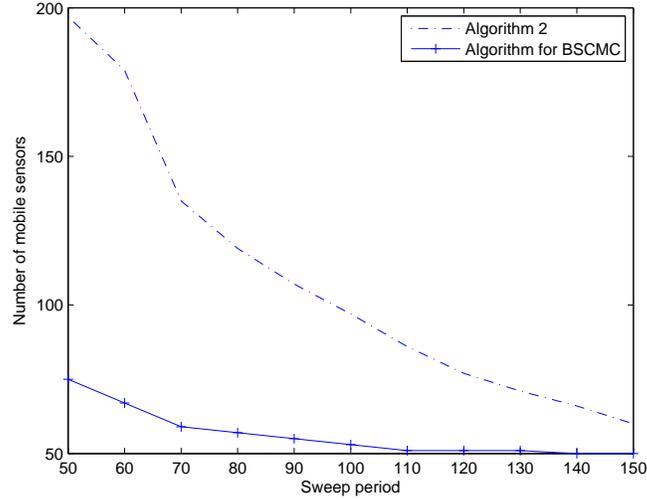}
    \caption{Comparison with respect to  number of mobile sensors varying with sweep period (in second)}\label{fig1:vs-t}
\end{figure}

Table \ref{tab1} shows comparison of average number of mobile sensors to achieve sweep coverage varying with the sweep periods.
The average number of mobile sensor is calculated for 100 several executions of the algorithms with fixed number of line segments,
which is equal to 50. A graphical representation of the Table \ref{tab1} is illustrated in Fig. \ref{fig1:vs-t}. Table \ref{tab1}
and Fig. \ref{fig1:vs-t} show that with increasing number of line segments, the difference between the average number of mobile
sensors decreases. In general the Algorithm for BSCMC performs better than the Algorithm \ref{alg:pint1}.

\section{Conclusion}\label{sec:concl}
Unlike traditional coverage, in sweep coverage periodic monitoring is
maintained by mobile sensors instead of continuous monitoring. There are many
applications in industry, where periodic monitoring is required for identification
of specific preventive maintenance. For example, periodic monitoring of electrical
equipments like motors and generators is required to check their partial discharges
\cite{Paoletti99}.

In this paper we have introduced sweep coverage concept for barriers, where the objective is to
cover finite length curves on a plane. In barrier sweep coverage, mobile sensors periodically visit
all points of a set of finite length curves. For a single curve, we have solved the problem optimally.
To resolve the issue of limited battery power of mobile sensors, we have proposed a solution by introducing
an energy source on the plane and proposed a solution, which achieves a constant approximation factor $\frac{13}{3}$.
We have proved that finding minimum number of mobile sensors to sweep cover a set of finite length curves is NP-hard and
cannot be approximated within a factor of 2. For this problem we have proposed a 2-approximation algorithm for a
special case, which is the best possible approximation factor. For the general problem, we propose a 5-approximation algorithm.
As an application of barrier sweep coverage problems, we have defined a data gathering problem with data mules,
where the concept of barrier sweep coverage is applied for gathering data by utilizing minimum number of data mules.
A 3-approximation algorithm is proposed to solve the problem. In future we want to investigate the sweep coverage
problems in presence of obstacles. There would be another possible extension of this work a plane, where the surface
might be uneven.
%Also we want to investigate the energy restricted version of barrier sweep coverage problem for multiple curves.
\bibliographystyle{plain}
\bibliography{bib1}

\end{document}